\documentclass[11pt]{article}%

\usepackage[ruled,vlined,linesnumbered]{algorithm2e}

\newcommand{\lncsstyle}[1]{}
\newcommand{\normalstyle}[1]{#1}

\normalstyle{
  \usepackage{fullpage}
}
\usepackage{amssymb,amsmath}
\usepackage{graphicx, epsfig}
\usepackage{tikz}

\tikzset{
   treenode/.style = {align=center, inner sep=0pt, text centered,
    font=\sffamily},
  int_m/.style = {treenode, circle, black, fill,font=\sffamily\bfseries, draw=black,
     text width=1.0em, very thick},
  int_u/.style = {treenode, circle, black, draw=black,  
    text width=1.0em, thick},
  leaf_m/.style={treenode, rectangle, black, fill, draw=black, 
    text width=1em,minimum width=1em, minimum height=1em}, 
  leaf_u/.style = {treenode, rectangle, black,draw=black,
    minimum width=0.9em, minimum height=0.9em},
  level 1/.style={level distance=15mm,sibling distance=60mm},
  level 2/.style={level distance=15mm,sibling distance=22mm},
  level 3/.style={level distance=15mm,sibling distance=12mm},
  level 4/.style={level distance=15mm,sibling distance=5mm}
}

\newcommand{\no}[1]{}
\newcommand{\myparagraph}[1]{{\bf #1}}
\newcommand{\todo}[1]{} 

\normalstyle{
\newtheorem{theorem}{Theorem}
\newtheorem{lemma}{Lemma}

\newenvironment{proof}{\trivlist\item[]\emph{Proof}:}%
{\unskip\nobreak\hskip 1em plus 1fil\nobreak$\Box$
\parfillskip=0pt%
\endtrivlist}
}

\newenvironment{itemize*}%
  {\begin{itemize}%
    \setlength{\itemsep}{0pt}%
    \setlength{\parskip}{0pt}%
    \setlength{\parsep}{0pt}%
    \setlength{\topsep}{0pt}%
    \setlength{\partopsep}{0pt}%
  }%
  {\end{itemize}}%

\renewenvironment{proof}{\trivlist\item[]\emph{Proof}:}%
{\unskip\nobreak\hskip 1em plus 1fil\nobreak$\Box$
\parfillskip=0pt%
\endtrivlist}

\newcommand{\cR}{{\cal R}}

\newcommand{\eps}{\varepsilon}

\newcommand{\sq}{\mathtt{square}}
\newcommand{\cell}{\mathtt{cell}}

\newcommand{\mleft}{\mathtt{left}}
\newcommand{\mright}{\mathtt{right}}
\newcommand{\mbot}{\mathtt{bot}}
\newcommand{\mtop}{\mathtt{top}}

\newcommand{\longver}[1]{}

\begin{document}
\title{Orthogonal Range Reporting and Rectangle Stabbing for Fat Rectangles}
\author{Timothy M. Chan\thanks{Department of Computer Science, University of Illinois at Urbana-Champaign   Email: {\tt tmc@illinois.edu}.} 
\and 
Yakov Nekrich \thanks{Cheriton School of Computer Science, University of Waterloo.
  Email: {\tt yakov.nekrich@googlemail.com}.} 
\and
Michiel Smid \thanks{School of Computer Science, Carleton University Email: {\tt michiel@scs.carleton.ca}.}
}
\date{}
\maketitle

\thispagestyle{empty}
\begin{abstract}
In this paper we study two geometric data structure problems in the special case when input objects or queries are fat rectangles.  We show that in this  case a significant improvement compared to the general case can be achieved. 

We describe  data structures that answer  two- and three-dimensional orthogonal range reporting queries in the case when the query range is a \emph{fat} rectangle.  Our two-dimensional data structure uses $O(n)$ words and supports queries in $O(\log\log U +k)$ time, where $n$ is the number of points in  the data structure,  $U$ is the size of the universe and $k$ is the number of points in the query range. 
 Our three-dimensional data structure needs $O(n\log^{\eps}U)$ words of space and answers queries in $O(\log \log U + k)$ time.  We also consider the rectangle stabbing problem on a set of three-dimensional fat rectangles. Our data structure uses $O(n)$ space and answers stabbing queries in $O(\log U\log\log U  +k)$ time.  
\end{abstract}
\newpage
\setcounter{page}{1}

\section{Introduction}
\label{sec:intro}


Orthogonal range reporting and rectangle stabbing are two fundamental problems in computational geometry. In the orthogonal range reporting problem we keep a set of points in a data structure; for any axis-parallel query rectangle $Q$ we must report all points in $Q$.  Rectangle stabbing is, in a sense, a dual problem. We keep a set of axis-parallel rectangles in a data structure. For a query point $q$ we must report all rectangles that are stabbed by $q$, i.e., all rectangles that contain $q$. A rectangle is fat if its aspect ratio (the ratio of its longest and shortest edges)  is bounded by a constant. In this paper we consider the range reporting problem in scenario when query rectangles are fat.  We show that  significant improvements can be achieved for this special case.  We also describe a data structure that supports  three-dimensional stabbing queries on a set of fat three-dimensional rectangles.


The range reporting problem and its variants have been studied extensively over the last four decades; see for example, ~\cite{GabowBT84,Chazelle86,Chaz88,VengroffV96,ArgeSV99,AlstrupBR00,Nekrich07,Nekrich07socg,Afshani08,KarpinskiN09,Chan13,ChanLP11}. We refer to~\cite{KreveldL2014,Nekrich2008} for extensive surveys of previous results.
The best known data structure for two-dimensional point reporting uses $O(n\log^{\eps}n)$ words of space and supports queries in $O(\log\log U + k)$ time~\cite{ChanLP11}. Henceforth $n$ is the total number of geometric objects (points or  rectangles) in the data structure, $k$ is the number of reported objects, and $\eps$ is an arbitrarily small positive constant; we assume that all point coordinates are positive integers bounded by a parameter $U$. The space usage can be reduced to linear or almost-linear at the cost of paying a non-constant penalty for every reported point. Thus there is an $O(n)$-word data structure that supports queries in $O(\log\log U + (k+1)\log^{\eps}n )$ time and $O(n \log\log n)$-word data structure that answers queries in $O(\log\log U + k\log\log n)$ time. If we want to use linear space and spend constant time for every reported point, then the overall query cost is increased to polynomial: the fastest linear-space data structure requires $O(n^{\eps}+k)$ time to answer a query~\cite{BentleyM80}. Better results are known only in the special case when the query range is bounded on three sides~\cite{McCreight85,AlstrupBR00}; there is a linear-space data structure that answers three-sided queries in $O(\log \log U+k)$ time (or even in $O(1+k)$ time if $U=O(n)$)~\cite{AlstrupBR00}. In this paper we show that two-dimensional orthogonal  range reporting queries can be answered in $O(\log\log U + k)$ time using an $O(n)$-space data structure under assumption that query rectangles are fat.  
 We also demonstrate that the fatness assumption is profitable for three-dimensional orthogonal range reporting.  We  show in this paper how to report all points in a three-dimensional axis-parallel fat rectangle in $O(\log\log U + k)$ time using a $O(n\log^{\eps} U)$-word data structure. This is to be compared with the result of Chan et al.~\cite{ChanLP11} that achieves the same query time for arbitrary rectangle queries but uses $O(n\log^{1+\eps}n)$ words of space. The third problem considered in this paper is the three-dimensional stabbing problem on a set of fat rectangles. For a query point $q$ we must report all rectangles that are stabbed by $q$. We describe a data structure that uses $O(n)$ words of space and supports queries in $O(\log U \log\log U + k)$ time.  For comparison, the best known data structures for general rectangles use $O(n\log^{*}n)$ words of space and support queries in $O(\log^2n)$ time~\cite{Rahul15}. 

 \begin{table}[tbh]
   \centering
   \begin{tabular}{|l|c|c|c|} \hline
     Reference & Space & Time & Query Type\\ \hline
     \cite{ChanLP11}   & $O(n)$ & $O(\log\log U + (k+1)\log^{\eps}n )$ & General\\
     \cite{ChanLP11}   & $O(n \log\log n)$ & $O(\log \log U + k\log\log n)$ & General\\       
     \cite{ChanLP11}    & $O(n\log^{\eps} n)$ & $O(\log\log U + k)$ & General\\
     \cite{BentleyM80}  & $O(n)$ & $O(n^{\eps} + k)$ & General\\
     \cite{McCreight85} & $O(n$  & $O(\log n + k)$ & Three-sided\\
     \cite{AlstrupBR00} & $O(n$  & $O(\log \log U + k)$ & Three-sided\\
     \cite{ChazelleE87} & $O(n)$ & $O(\log n + k)$ & Fat\\
     {\bf New}          & $O(n)$ & $O(\log\log U + k)$ & Fat \\\hline
    \end{tabular}
   \caption{Space-time trade-offs for two-dimensional range reporting. Result in line 7 is a corollary from~\cite{ChazelleE87}, but it is not stated there.}
   \label{tab:res2dim}
 \end{table}

Our data structure for two-dimensional range reporting, described in Sections~\ref{sec:subdivision} and~\ref{sec:fatrangerep}, is based on quadtrees. Using a marking scheme on nodes of a quadtree, we divide the plane into $O(n/d)$ canonical rectangles, so that each rectangle contains $O(d)$ points for $d=\log n$. For any fat query rectangle $Q$, we can quickly find  all canonical rectangles $R$ satisfying $Q\cap R\not=\emptyset$ and report all points in $Q\cap R$ for all such $R$.  In Section~\ref{sec:orthrep3d} we describe a data structure that supports three-dimensional range reporting for fat query ranges. It is based on recursive decomposition of the grid similar to \cite{AlstrupBR00,ChanLP11,KarpinskiN09}, but in our case the grid is divided into uniform cells.  We describe a data structure for stabbing queries on a set of three-dimensional fat rectangles in Section~\ref{sec:stab3d}. This result is based on reducing a  stabbing query to $O(\log U)$ three-dimensional dominance queries.  The results of this paper are valid in the word RAM model of computation. 

\myparagraph{Related Work.}
A result about range reporting in two-dimensional fat rectangles is implicitly contained in the paper of Chazelle and Edelsbrunner~\cite{ChazelleE87}. In~\cite{ChazelleE87} the authors describe a linear-space data structure for triangular range reporting. Their data structure can report all points in an arbitrary  query  triangle, provided that the  sides of the triangle parallel to three fixed directions; queries are supported in  $O(\log n+ k)$ time where $k$ is the number of reported points. We can represent a square as a union of two such triangles and we can represent an arbitrary fat rectangle  as a union of $O(1)$ squares. Hence we can answer two-dimensional range reporting queries for fat rectangles in $O(\log n+ k)$ time and $O(n)$ space. 

Data structures for fat convex objects are studied in e.g.,~\cite{Katz97,EfratKNS00,BergG08,AronovBG08}. Iacono and Langerman~\cite{IaconoL00} describe a data structure that supports point location queries in a set of axis-parallel fat $d$-dimensional rectangles. This data structure answers queries in $O(\log\log U)$ time and uses $O(n \log\log U)$ space for any fixed dimension $d$.

\section{Quadtree-Based Rectangular Subdivision}
\label{sec:subdivision}
In this section we describe a planar rectangular subdivision that is used by our two-dimensional data structure. To make the description self-contained, we start with the definition of a compressed quadtree. 

A quadtree $T_Q$ is a hierarchical data structure that divides the plane into regions. Let $U$ denote the maximum of $x$- and $y$-coordinates of all points. 
We associate a square (also called a cell) $\sq(v)$ to every quadtree node $v$. The root of a quadtree is associated to  the square $[0,U]\times [0,U]$. W. l. o. g. we assume that $U$ is a power of $2$. If a square $\sq(v)$ of a node $v$ contains more than one point, then the node $v$ has four children. We divide $\sq(v)$  into four squares of equal size and  associate them  to four child nodes of $v$. A compressed quadtree $T$ is a subtree of $T_Q$ obtained  by keeping only those internal nodes of $T_Q$ that have more than one non-empty child.

\myparagraph{Marking Nodes in a Quadtree.}
Let $T$ denote a compressed quadtree on a set of $n$ points. Let $d=\log n$. We mark selected nodes in $T$  by employing the following marking scheme: (i) every $d$-th leaf is marked and (ii) if an  internal node $u$  has at least two children with marked descendants, then $u$  is marked.  We can mark nodes of a given quadtree $T$ in linear time using the following method. We will say that a node $u$ is a special node if exactly one child of $u$ has marked descendants. First we traverse the leaves of $T$ in the left-to-right order and mark every $d$-th leaf, starting with the leftmost one. Then we visit all internal nodes of $T$ in post-order. If a visited node $u$ has exactly one child $u_i$ such that $u_i$ is either marked or special, then we declare that the node $u$ is special. If $u$ has two or more children that are either special or marked, then the node $u$ is marked. Marked nodes induce a subtree $T'$ of $T$. $T'$ has $n/d$ leaves. Since  every internal node of $T'$ has at least two children, $T'$ has at most $n/d -1$ internal nodes. Hence the total number of marked nodes is $O(n/d)$.  Similar methods for selecting nodes were previously used in other tree-based data structures, see e.g.,~\cite{NavarroN12,LewensteinNV14}.

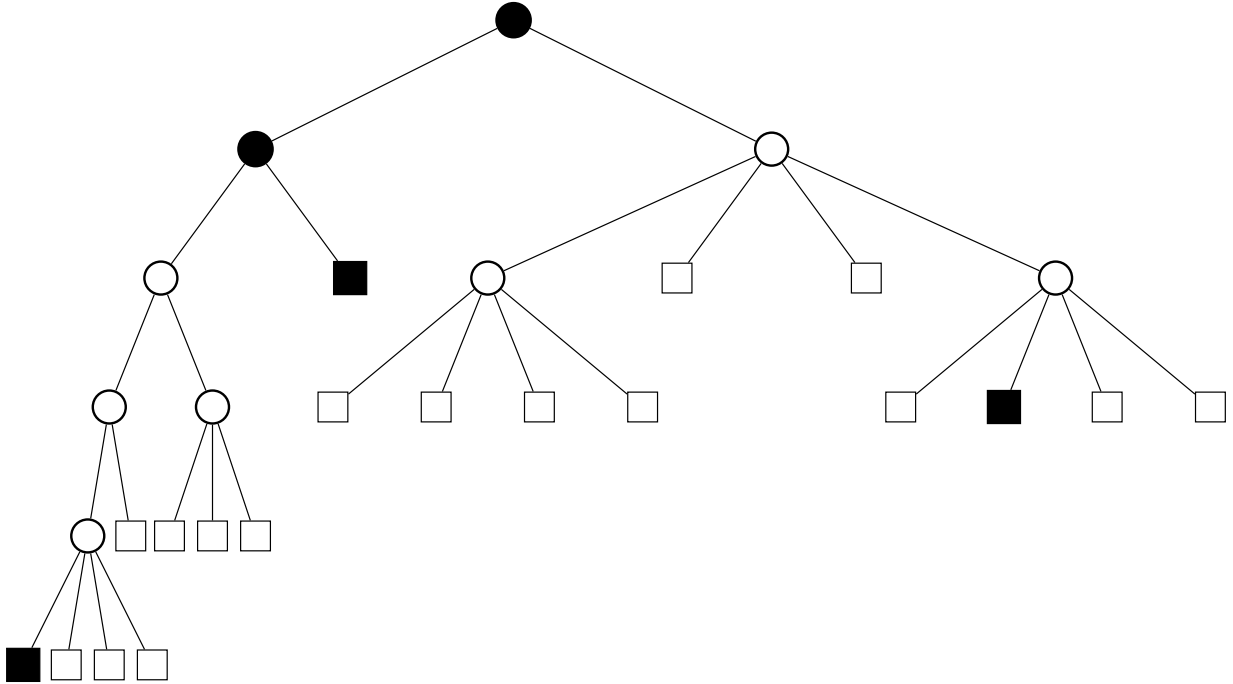
\begin{figure}[tb]
  \centering
\resizebox{\textwidth}{!}{
  \begin{tikzpicture}
    \node[int_m]{}
    child{ node[int_m]{}
      child{ node[int_u]{}      
        child{ node[int_u]{}
          child{ node[int_u]{}
            child{ node[leaf_m]{}}
            child{ node[leaf_u]{}}
            child{ node[leaf_u]{}}
            child{ node[leaf_u]{}}
          }
          child{ node[leaf_u]{}}
        }
        child{ node[int_u]{}
          child{ node[leaf_u]{}}
          child{ node[leaf_u]{}}
          child{ node[leaf_u]{}}
        } 
      }
      child{ node[leaf_m]{} 
    }
  }
  child{node[int_u]{}
    child{ node[int_u]{}
      child{ node[leaf_u]{}}
      child{ node[leaf_u]{}}
      child{ node[leaf_u]{}}
      child{ node[leaf_u]{}}
      }
      child{ node[leaf_u]{}}
      child{ node[leaf_u]{}}
      child{ node[int_u]{}
        child{ node[leaf_u]{}}
        child{ node[leaf_m]{}}
        child{ node[leaf_u]{}}
        child{ node[leaf_u]{}}
        }
  };
  \end{tikzpicture}
}
  \caption{Marking nodes in a compressed quadtree for  $d=8$. Leaves are shown with squares and internal nodes are shown with circles. Marked leaves and internal nodes are depicted with filled circles and filled squares respectively. Only a part of the quadtree is shown.}
  \label{fig:quadmark}
\end{figure}

\myparagraph{Rectangular Subdivision.} 
When nodes are marked,  we traverse $T$ from the top to the bottom and divide it into  $O(n/d)$ rectangles so  that each rectangle contains  $O(d)$ points. The subdivision is produced as follows. A direct marked descendant of a node $u$ is a descendant $u'$ of $u$ such that $u'$ is marked and there are no marked nodes between $u$ and $u'$. Suppose that a node $u$ is a marked node and let $u_1$, $\ldots$, $u_f$ denote its direct marked descendants. A marked node has at most $4$ direct marked descendants, therefore $f\le 4$. Let $\sq(u)$ denote the cell of a node $u$. We can represent $\sq(u)\setminus (\cup_{i=1}^f \sq(u_i))$ as a union of a constant number of rectangles $R_j(u)$. We will say that $R_j(u)$ are rectangles associated to  the node $u$.  See Fig.~\ref{fig:markdesc}.  
\begin{figure}[tb]
  \centering
  \includegraphics[width=.4\textwidth,page=2]{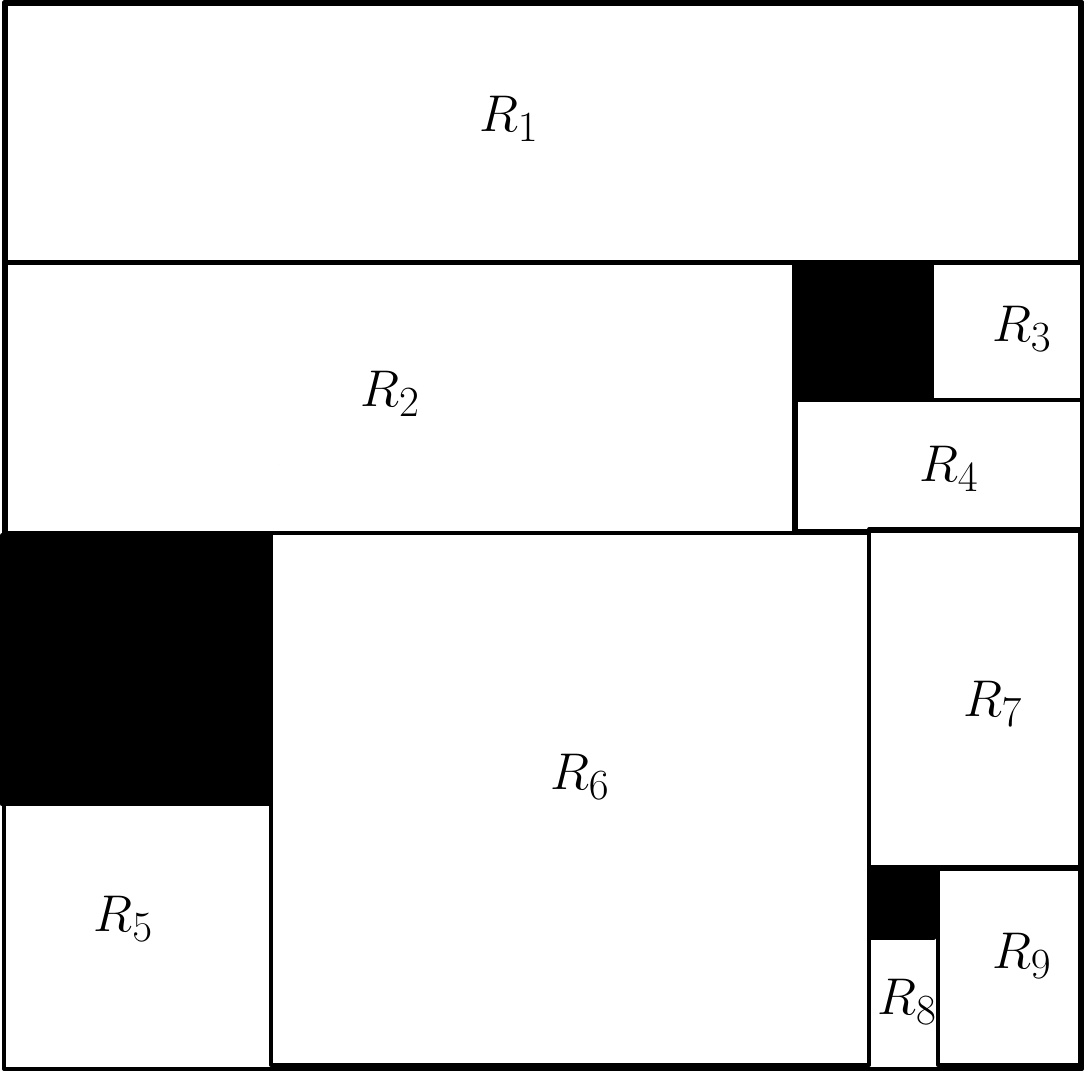}
  \caption{Subdivision of a marked cell into rectangles. Cells corresponding to direct marked descendants are shown in black. }
  \label{fig:markdesc}
\end{figure}
There are $O(n/d)$ marked nodes in the quadtree. Our subdivision consists of rectangles $R_i(u)$ for all marked internal nodes of $T$ and cells $\sq(v)$ for all marked leaves $v$ of $T$. By dividing every marked node with marked descendants into rectangles as described above, we obtain a sub-division of the plane into $O(n/d)$ rectangles. Rectangles of this subdivision will be further called \emph{canonical rectangles}.

\begin{lemma}
\label{lemma:canrect}
  Every canonical rectangle contains $O(d)$ points. 
\end{lemma}
\begin{proof}
  Consider a rectangle $R(u)$ associated to a node $u$. Let $u_1$, $\ldots$, $u_f$ denote the direct marked descendants of $u$. We can show that the set $P_0=\sq(u)\setminus (\cup_{i=1}^f \sq(u_i))$ contains $O(d)$ points. Let $L_0$ denote the set of leaves in which points from $P_0$ are stored. There are at most $d$ leaf nodes from $L_0$ between $u_i$ and $u_{i+1}$ for $1\le i <f$; there are at most $d$ leaf descendants of $u$ to the left of $u_1$ and at most $d$ leaf descendants of $u$ to the right of $u_f$.  Hence the total number of leaves in $L_0$ does not exceed $(f+2)d$. Since $R(u)\subseteq L_0$ and $f\le 4$, $R(u)$ contains $O(d)$ points. 
\end{proof}


\section{Orthogonal Range Reporting for Fat Boxes in 2-D}
\label{sec:fatrangerep}

\myparagraph{Data Structure.}
We divide the plane into canonical rectangles as described in Section~\ref{sec:subdivision}. For every rectangle $R$  in this subdivision we keep the list $L_x(R)$ of points in $R$ sorted by their $x$-coordinates  and the list $L_y(R)$ of points in $R$ sorted by their $y$-coordinates. We also keep a data structure $D(R)$ that supports two-dimensional range reporting queries on points of $R$. Since $R$ contains $O(\log n)$ points, we can implement $D(R)$ in $O(\log n)$ space  so that queries are  supported in $O(k)$ time. The data structure $D(R)$ will be described in Section~\ref{sec:small}. We will denote by $P$ the set of points stored in our data structure.

\myparagraph{Orthogonal Range Queries.}
For simplicity we will consider the case when the query range is a square. Any fat rectangle can be represented as a union of $O(1)$ squares. Consider a query $Q=[a,b]\times [c,d]$. All canonical rectangles that intersect $Q$ can be divided into three categories: (i) corner rectangles that contain a corner of $Q$ (ii) rectangles that cut one side of $Q$ or are completely contained in $Q$; such rectangles will be called internal rectangles (iii) rectangles that cross two opposite sides of $Q$, but do not contain corners of $Q$; we say that such rectangles are spanning rectangles or that type (iii) rectangles span $Q$. See Fig.~\ref{fig:rect}. 
\begin{figure}[tb]
  \centering
    \includegraphics[width=.35\textwidth]{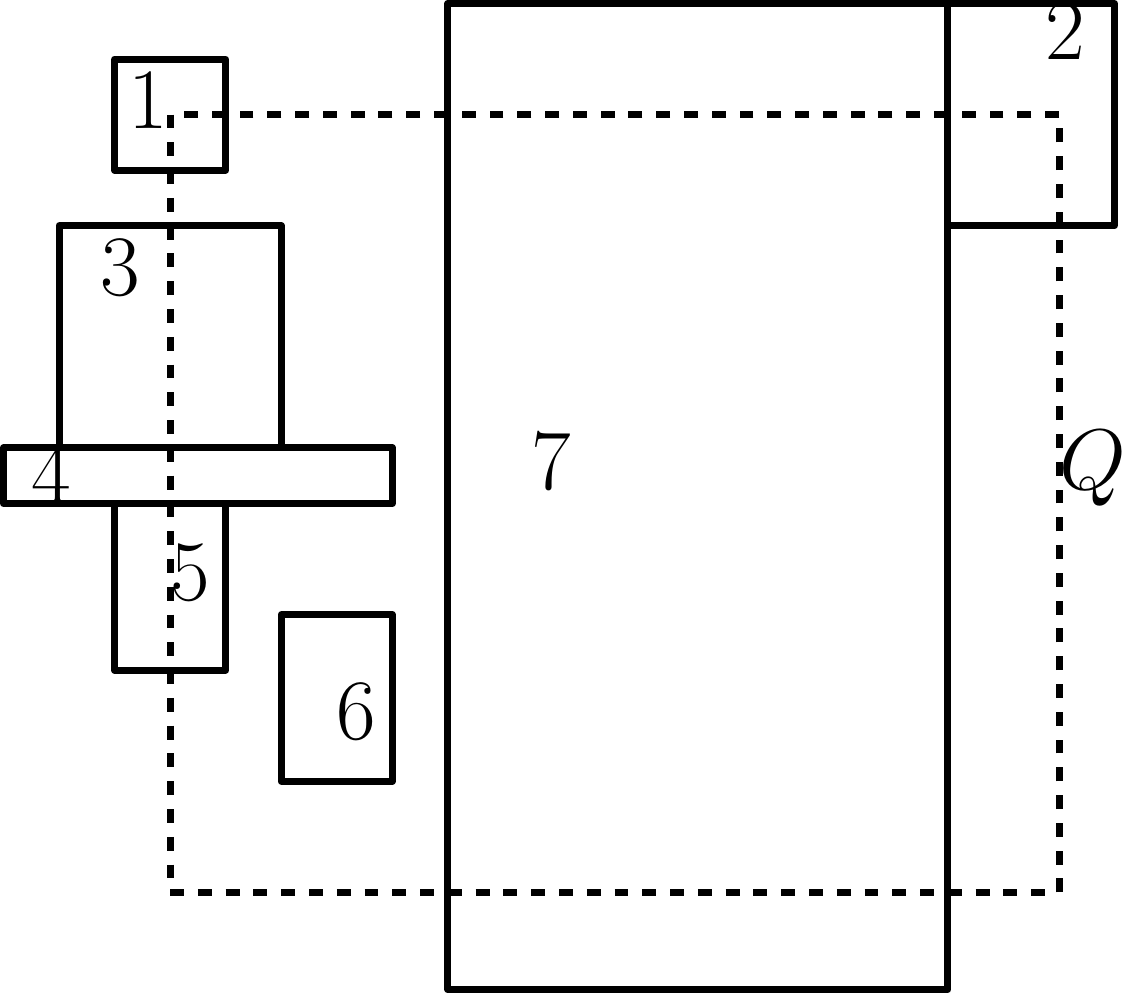}
  \caption{Examples of different rectangles with respect to a query $Q$. Rectangles 1 and 2 are corner rectangles, rectangles 3, 4, 5, 6 are internal rectangles, and rectangle 7 spans $Q$.}
  \label{fig:rect}
\end{figure}

\begin{lemma}
\label{lemma:span}
  If a rectangle $Q=[a,b]\times [c,d]$ is a square, then $Q$ is spanned by $O(1)$ canonical rectangles. 
\end{lemma}
\begin{proof}
Suppose that a canonical rectangle $R(u)$, associated to a node $u$,  spans $Q$. Then either (i) $\sq(u)$ contains two corners of $Q$ and $Q$ is not contained in $\sq(u')$ for any descendant $u'$ of $u$, or (ii) $\sq(u)$ contains $Q$ but $Q$ is not contained in $\sq(u')$ for any descendant $u'$ of $u$.  

If $Q$ is contained in $\sq(u)$ and at least one  rectangle $R(u)$ spans $Q$, then $Q$ is 
not contained in $\sq(u')$ for any descendant $u'$ of $u$. Hence there is at most one cell that satisfies condition (i).  

Suppose that $\sq(u)$ contains two corners of $Q$ and some rectangle $R(u)$ spans $Q$. Let us assume w.l.o.g. that $Q$ crosses the left side $\ell$ of $\sq(u)$. Let $u'$ be some descendant of $u$. If $\sq(u')$ for a descendant $u'$ of $u$ does not touch the left side of $\sq(u)$, then the distance from $\ell$ to $\sq(u')$ is greater than or equal to  the size of $\sq(u')$. Hence $\sq(u')$ does not contain two corners of $Q$. If $\sq(u')$ touches $\ell$ and contains two corners of $Q$, then there is no canonical rectangle $R(u)$ that spans $Q$. Thus there is at most one cell  that satisfies condition (ii).

Since there is only one cell satisfying condition (i) and only one 
cell satisfying condition (ii), the total number of canonical rectangles that span $Q$ is bounded by a constant. 

\end{proof}

A query range can overlap with a large number of internal rectangles. But we can find all internal rectangles $R$, such that $R\cap Q\cap P\not=\emptyset$ by answering a range reporting query on a set $P'$ (defined below) that contains $O(n/d)$ representative points for $d=\log n$.  It was shown in Lemma~\ref{lemma:span} that a   square  is intersected by $O(1)$ spanning rectangles. There are at most four corner rectangles for any query range $Q$.  Since there is a constant number of corner rectangles and spanning rectangles, we can process all of them in constant time. A more detailed description follows.

We can identify all internal rectangles  (type (ii) rectangles) that contain at least one point from $P\cap Q$ as follows.  For every canonical rectangle, we keep its topmost point, its lowermost point, its leftmost point, and its rightmost point in the set $P'$. $P'$ contains $O(n/d)$ points. We keep $P'$ in the data structure $D'$ that supports orthogonal range reporting queries in $O(\log\log U + k)$ time~\cite{ChanLP11}. $D'$ uses space $O(n' \log^{\eps}n')$, where $n'$ is the number of points in $P'$. Since $n'=O(n/d)$, $D'$ uses space $O(n)$. If $rect(u)$ is an internal  rectangle  and  $rect(u)\cap Q\cap P\not=\emptyset$, then at least one of its extreme points is in $Q$. We can find all such rectangles by answering the same query $Q$ on the set $P'$. For every reported point $p$, we examine the canonical rectangle $R_p$ that contains $p$.  

There are at most four corner rectangles.  We can find corner rectangles  by keeping all canonical rectangles in the point location data structure of Chan~\cite{Chan13}. For each corner point $q$ of $Q$, we identify the rectangle $R_q$ that contains $q$ in $O(\log\log U)$ time.

Rectangles that span $Q$ are the most difficult to deal with. All  points of a spanning rectangle can be outside of $Q$. It is not clear how we can find spanning rectangles $R$ such that $R\cap Q\cap P \not=\emptyset$.   Existence of these rectangles is the reason why our method cannot be extended to the general case of the orthogonal range reporting.  However, by Lemma~\ref{lemma:span}, a square query range $Q$  is spanned by $O(1)$ canonical rectangles from our subdivision.  
All rectangles that span $Q$ can be found as follows. For a rectangle $R$ we denote by $\mleft(R)$, $\mright(R)$, $\mbot(R)$, and $\mtop(R)$ the lower and upper bounds of its horizontal and vertical projections; that is, $R= [\mleft(R),\mright(R)]\times [\mbot(R),\mtop(R)]$. If a rectangle $R$ spans $Q$, then at least one side of $R$ spans $Q$. That is, $R$ satisfies one of the following conditions:  (i) $\mleft(R)\le a$, $\mright(R)\ge b$, and $c\le \mtop(R)\le d$; (ii) $\mleft(R)\le a$, $\mright(R)\ge b$, and $c\le \mbot(R)\le d$; (iii) $a\le \mleft(R)\le b$, $\mbot(R)\le c$, and $\mtop(R)\ge d$; (iv) $a\le \mright(R)\le b$, $\mbot(R)\le c$, and $\mtop(R)\ge d$. 
We keep information about every rectangle in four three-dimensional data structures. The data structure $\cR_1$ contains a tuple 
$(\mleft(R),\mright(R), \mtop(R))$ for every canonical rectangle $R$. $\cR_1$ can find all $R$ that satisfy  $\mleft(R)\le a$, $\mright(R)\ge b$, and $c\le \mtop(R)\le d$.  The data structure $\cR_2$ contains a tuple 
$(\mleft(R),\mright(R), \mbot(R))$ for every canonical rectangle $R$. $\cR_2$ can find all $R$ that satisfy  $\mleft(R)\le a$, $\mright(R)\ge b$, and $ c \le \mbot(R)\le d$. Data structures $\cR_3$ and $\cR_4$ contain tuples $(\mleft(R), \mbot(R),\mtop(R))$ and $(\mright(R),\mbot(R),\mtop(R))$ respectively for every canonical rectangle $R$. $\cR_3$ supports queries $a\le \mleft(R)\le b$, $\mbot(R)\le c$, and $\mtop(R)\ge d$; $\cR_4$ supports queries $a\le \mright(R)\le b$, $\mbot(R)\le c$, and $\mtop(R)\ge d$. 
Queries supported by data structures $\cR_i$ are a special case of three-dimensional orthogonal range reporting queries, called 4-sided queries (the query range is bounded on four sides). Using the result of Chan et al.~\cite{ChanLP11}, we can answer such queries in $O(\log\log U+k)$ time using $O(n'\log^{\eps}n)$ space where $n'=O(n/d)$ is the number of tuples in $\cR_i$. If a rectangle $R$ is returned by a query to $\cR_i$, then $R$ spans $Q$ or $R$ contains two corners of $Q$.  If $Q$ is a square, then we can answer all queries on $\cR_i$ described above and identify all canonical rectangles that span $Q$ in $O(\log\log U + f)=O(\log\log U)$ time, where $f=O(1)$ is the number of canonical rectangles that span $Q$. 

For every corner or  spanning rectangle $R$,  we find all points in $R\cap Q$ using the data structure $D(R)$. Since the total number of corner and spanning rectangles is bounded by $O(1)$, we can find all relevant points in $O(k)$ time.  Using data structure $D'$ we can find all internal rectangles in  $O(\log\log U +n_I)$ time where $n_I$ is the number of internal rectangles. For every internal rectangle $R_I$ we traverse the list of points in $L_x(R_I)$ or $L_y(R_I)$ and report all points in $R_I\cap Q$ in time $O(k_I)$ where $k_I=|R_I\cap Q|$. 
The result of this section can be summed up as follows.
\begin{theorem}
  \label{theor:fatrangerep}
There is a linear-space data structure that answers orthogonal range reporting queries in $O(\log\log U + k)$ time provided the query range $Q=[a,b]\times [c,d]$ is a fat rectangle. 
\end{theorem}

\section{Orthogonal Range Reporting for Fat Boxes in 3-D}
\label{sec:orthrep3d}
In this section, we describe a data structure for 3-d orthogonal
range reporting for fat query boxes, by adopting
a recursive grid approach.  Nonuniform grids have been
used in previous range searching data structures
by Alstrup, Brodal, and Rauhe~\cite{AlstrupBR00} and Chan,
Larsen, and Patrascu~\cite{ChanLP11}, but we use uniform grids
instead.  Also, the way we use
recursion is a little different, and
more closely resembles the recursion from van Emde
Boas trees.  Each node in our recursive structure is
augmented with a general 5-sided range reporting structure;
thus, our solution can be viewed as a reduction from fat
6-sided range searching to 5-sided range searching.

\myparagraph{The data structure.}
Let $P$ be a given set of $n$ points in $[U]^3$, where $[U]$ denotes $\{0,1,\ldots,U-1\}$.
Let $r$ be a parameter  (a function of $U$) to be chosen later.
Divide $[U]^3$ into $r^3$ \emph{grid cells}, each  a cube 
of the form $\{(x,y,z) : (U/r)i \le x < (U/r)(i+1),\ 
(U/r)j \le y < (U/r)(j+1),\
(U/r)k \le z < (U/r)(k+1)\}$ for some $i,j,k\in [r]$.
We call $(i,j,k)$ the label of such a grid cell.
A \emph{grid slab} refers to a region of the form $\{(x,y,z):
(U/r)i \le x < (U/r)(i+1)\}$, $\{(x,y,z):
(U/r)j \le y < (U/r)(j+1)\}$, or $\{(x,y,z):
(U/r)k \le z < (U/r)(k+1)\}$.
A \emph{grid-aligned box} refers to a box whose $x$-, $y$-,
and $z$-coordinates are all multiples of $U/r$.
We construct our data structure as follows:
\begin{enumerate}
\item[A.] For each nonempty grid cell $\gamma$,
recursively build a data structure for $P\cap\gamma$; also
store $P\cap\gamma$ as a linked list.
\item[B.] Let $\Gamma$ be the set of all nonempty grid cells.
Recursively build a data structure for the labels of $\Gamma$.
\item[C.] For each grid slab $\sigma$,
build Chan, Larsen, and Patrascu's data structure~\cite{ChanLP11} for $P\cap\sigma$ for 3-d 5-sided queries, which requires
$O(n\log^\eps n)$ words of space and $O(\log\log U)$ query time\footnote{For simplicity, we ignore the time needed to output points in this section.}.
\end{enumerate}

\myparagraph{Analysis of space.}
Since we use a uniform grid, we will represent the space usage and query time as functions of the universe size $U$.
Let $s(U)$ be the \emph{amortized} space complexity of our data structure in bits, i.e., the total space complexity in bits
divided by the number of points $n$.
Item~A of the data structure requires at most $s(U/r)$ bits per point, since after translation, each grid cell becomes $[U/r]^3$.
This ignores the space for the linked lists, which require a total of $O(n\log U)$ bits.
Item~B requires at most $s(r)$ bits per point, since the labels lie in $[r]^3$.
Item~C requires a total of $O(n\log^\eps n\log U)\le O(n\log^{1+\eps} U)$ bits (since $n\le U^3$).
Thus,
\[ s(U) \ \le\ s(U/r) + s(r) + O(\log^{1+\eps} U).
\]

\myparagraph{Query algorithm.}
We consider the case when the query range is an (axis-parallel) cube; any
fat query box can be expressed as a union of $O(1)$ cubes.
Given a query cube $Q$, we report all points of $P$ in $Q$ as follows:
\begin{enumerate}
\item If $Q$ is completely contained in a grid cell $\gamma$, then
recursively report all points of $P\cap\gamma$ in $Q$.  Otherwise:
\item Decompose $Q$ into (at most) one grid-aligned cube $Q'$ and (at most)
six other boxes $Q_1,\ldots,Q_6$, where each $Q_i$ is a 5-sided
box in a grid slab $\sigma_i$.  (See Figure~\ref{fig1} for an analogous 2-d depiction.)
\item Recursively report all grid cells of $\Gamma$ in $Q'$.
For each reported grid cell $\gamma\in \Gamma$, report all points in the linked list $P\cap \gamma$.
\item For each $i\in\{1,\ldots,6\}$, report all points
of $P\cap\sigma_i$ in $Q_i$.
\end{enumerate}

\begin{figure}
\begin{center}
\includegraphics[scale=1.5]{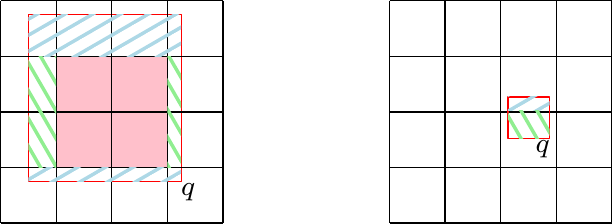}
\end{center}
\caption{In 2-d, if a query square $Q$ is not completely contained 
in any grid cell, it can be decomposed
into one grid-aligned square and four 3-sided rectangles
in four grid slabs (shown in the left), or just two 3-sided rectangles (shown in the right).  Similarly, in 3-d,
if a query cube $Q$ is not completely contained 
in any grid cell, it can be decomposed
into one grid-aligned cube and six 5-sided rectangles
in six grid slabs, or just two 5-sided rectangles.
}
\label{fig1}
\end{figure}

\myparagraph{Analysis of query time.}
Let $t(U)$ denote the running time of the query algorithm, excluding the outputting cost (which is $O(k)$ for $k$ output points).
Step~1 takes $t(U/r)$ time.
The recursive call in step~3 takes $t(r)$ time.
Step~2 takes $O(\log\log U)$ time.
Thus,
\[ t(U)\ \le\ \max\left\{t(U/r) + O(1),\ t(r) + O(\log\log U)\right\}.
\]

We can eliminate the $O(1)$ in the first term of the max by the following idea:
Consider the tree formed by expanding the recursion
due to item~A (treating the recursive structures from item~B as secondary structures at the nodes of the main tree).
Then we can jump to the first
node of the tree at which $Q$ is not completely contained in 
a grid cell, in $O(1)$ time by an LCA operation in the tree (actually, because the tree is perfectly balanced,
for our choice of $r$ (see below), the LCA operation can be simulated by standard arithmetic and bitwise-logical operations on the coordinates of~$Q$).

\myparagraph{Conclusion.}
Setting $r=U^{1/b}$ for a fixed parameter $b$ gives
\begin{eqnarray*}
s(U)  &\le& s(U^{1-1/b}) + s(U^{1/b}) + O(\log^{1+\eps} U)\\
t(U)  &\le& \max\left\{t(U^{1-1/b}),\ t(U^{1/b}) + O(\log\log U)\right\},
\end{eqnarray*}
which solves to $s(U)=O(b\log^{1+\eps} U)$ and
$t(U)=O(\log_b\log U\cdot\log\log U)$.  The outputting cost
goes up to $O(k\log_b\log U)$, since each point may be reported
$O(\log_b\log U)$ times.

Setting $b=\log^\eps U$ gives $O(\log^{1+2\eps}U)$ amortized space in bits and $O(\log\log U+k)$ query time.  Readjusting $\eps$ by a half,
we conclude: 
\begin{theorem}
We can store $n$ points in $[U]^3$ in a data structure with
$O(n\log^\eps U)$ words of space so that we can report
all $k$ points in any query fat box in $O(\log\log U + k)$ time.
\end{theorem}

\myparagraph{Remark.}
The above theorem can't be improved with current
state of the art, because
3-d 4-sided range reporting reduces to our problem and
the current best data structure for the former requires
$O(n\log^\eps n)$ space and $O(\log\log U + k)$ query time.
To see the reduction, note that
a 4-sided box $[x_1,x_2)\times (-\infty,y)\times (-\infty,z)$
contains the same points as the cube
$[x_1,x_2)\times [y-(x_2-x_1),y)\times [z-(x_2-x_1),z)$,
assuming that $x_2-x_1>U$.  The assumption can be guaranteed
after stretching the $x$-axis by a factor of $U$ (so that
the universe is now $[U^2]\times [U]\times [U]$).

\bibliographystyle{abbrv}
\bibliography{stretch}

\begin{thebibliography}{10}

\bibitem{Afshani08}
P.~Afshani.
\newblock On dominance reporting in 3d.
\newblock In {\em Proc.\ 16th Annual European Symposium on Algorithms({ESA})},
  pages 41--51, 2008.

\bibitem{AlstrupBR00}
S.~Alstrup, G.~S. Brodal, and T.~Rauhe.
\newblock New data structures for orthogonal range searching.
\newblock In {\em Proc. 41st Annual Symposium on Foundations of Computer
  Science, ({FOCS} 2000)}, pages 198--207, 2000.

\bibitem{ArgeSV99}
L.~Arge, V.~Samoladas, and J.~S. Vitter.
\newblock On two-dimensional indexability and optimal range search indexing.
\newblock In {\em Proc. 18th ACM SIGACT-SIGMOD-SIGART Symposium on Principles
  of Database Systems (PODS)}, pages 346--357, 1999.

\bibitem{AronovBG08}
B.~Aronov, M.~de~Berg, and C.~Gray.
\newblock Ray shooting and intersection searching amidst fat convex polyhedra
  in 3-space.
\newblock {\em Comput. Geom.}, 41(1-2):68--76, 2008.

\bibitem{BentleyM80}
J.~L. Bentley and H.~A. Maurer.
\newblock Efficient worst-case data structures for range searching.
\newblock {\em Acta Inf.}, 13:155--168, 1980.

\bibitem{Chan13}
T.~M. Chan.
\newblock Persistent predecessor search and orthogonal point location on the
  word {RAM}.
\newblock {\em {ACM} Trans. Algorithms}, 9(3):22, 2013.

\bibitem{ChanLP11}
T.~M. Chan, K.~G. Larsen, and M.~Patrascu.
\newblock Orthogonal range searching on the {RAM}, revisited.
\newblock In {\em Proc. 27th ACM Symposium on Computational Geometry, (SoCG
  2011)}, pages 1--10, 2011.

\bibitem{Chazelle86}
B.~Chazelle.
\newblock Filtering search: a new approach to query-answering.
\newblock {\em SIAM J. Comput.}, 15(3):703--724, 1986.

\bibitem{Chaz88}
B.~Chazelle.
\newblock A functional approach to data structures and its use in
  multidimensional searching.
\newblock {\em SIAM J. Comput.}, 17(3):427--462, 1988.

\bibitem{ChazelleE87}
B.~Chazelle and H.~Edelsbrunner.
\newblock Linear space data structures for two types of range search.
\newblock {\em Discrete {\&} Computational Geometry}, 2:113--126, 1987.

\bibitem{BergG08}
M.~de~Berg and C.~Gray.
\newblock Vertical ray shooting and computing depth orders for fat objects.
\newblock {\em {SIAM} J. Comput.}, 38(1):257--275, 2008.

\bibitem{EfratKNS00}
A.~Efrat, M.~J. Katz, F.~Nielsen, and M.~Sharir.
\newblock Dynamic data structures for fat objects and their applications.
\newblock {\em Computational Geometry}, 15(4):215 -- 227, 2000.

\bibitem{GabowBT84}
H.~N. Gabow, J.~L. Bentley, and R.~E. Tarjan.
\newblock Scaling and related techniques for geometry problems.
\newblock In {\em Proc.\ 16th Annual ACM Symposium on Theory of Computing (STOC
  1984)}, pages 135--143, 1984.

\bibitem{IaconoL00}
J.~Iacono and S.~Langerman.
\newblock Dynamic point location in fat hyperrectangles with integer
  coordinates.
\newblock In {\em Proc. 12th Canadian Conference on Computational Geometry},
  2000.

\bibitem{KarpinskiN09}
M.~Karpinski and Y.~Nekrich.
\newblock Space efficient multi-dimensional range reporting.
\newblock In {\em Proc. 15th Annual International Conference on Computing and
  Combinatorics ({COCOON} 2009)}, pages 215--224, 2009.

\bibitem{Katz97}
M.~J. Katz.
\newblock 3-d vertical ray shooting and 2-d point enclosure, range searching,
  and arc shooting amidst convex fat objects.
\newblock {\em Computational Geometry}, 8(6):299 -- 316, 1997.

\bibitem{KreveldL2014}
M.~v. Kreveld and M.~L{\"o}ffler.
\newblock {\em Range Searching}, pages 1--7.
\newblock Springer Berlin Heidelberg, Berlin, Heidelberg, 2014.

\bibitem{LewensteinNV14}
M.~Lewenstein, Y.~Nekrich, and J.~S. Vitter.
\newblock Space-efficient string indexing for wildcard pattern matching.
\newblock In {\em Proc.\ 31st International Symposium on Theoretical Aspects of
  Computer Science ({STACS} 2014)}, pages 506--517, 2014.

\bibitem{McCreight85}
E.~M. McCreight.
\newblock Priority search trees.
\newblock {\em {SIAM} J. Comput.}, 14(2):257--276, 1985.

\bibitem{NavarroN12}
G.~Navarro and Y.~Nekrich.
\newblock Top-\emph{k} document retrieval in optimal time and linear space.
\newblock In {\em Proc.\ 23rd Annual {ACM-SIAM} Symposium on Discrete
  Algorithms ({SODA} 2012}, pages 1066--1077, 2012.

\bibitem{Nekrich07socg}
Y.~Nekrich.
\newblock A data structure for multi-dimensional range reporting.
\newblock In {\em Proc.\ 23rd {ACM} Symposium on Computational Geometry,
  ({SoCG})}, pages 344--353, 2007.

\bibitem{Nekrich07}
Y.~Nekrich.
\newblock Space efficient dynamic orthogonal range reporting.
\newblock {\em Algorithmica}, 49(2):94--108, 2007.

\bibitem{Nekrich2008}
Y.~Nekrich.
\newblock {\em Orthogonal Range Searching on Discrete Grids}, pages 1--6.
\newblock Springer US, Boston, MA, 2008.

\bibitem{Rahul15}
S.~Rahul.
\newblock Improved bounds for orthogonal point enclosure query and point
  location in orthogonal subdivisions in $\mathbb{R}^3$.
\newblock In {\em Proc.\ 26th Annual {ACM-SIAM} Symposium on Discrete
  Algorithms ({SODA} 2015)}, pages 200--211, 2015.

\bibitem{VengroffV96}
D.~E. Vengroff and J.~S. Vitter.
\newblock Efficient 3-d range searching in external memory.
\newblock In {\em Proc.\ 28th Annual {ACM} Symposium on the Theory of Computing
  ({STOC} 1996)}, pages 192--201, 1996.

\end{thebibliography}

\appendix

\section{Orthogonal Range Reporting on a Small Set of Points}
 \label{sec:small}
In this section we show how two-dimensional orthogonal range reporting on a set of $d=O(\log n)$ points can be supported in $O(k)$ time. Our data structure
 uses  space $O(d)$, but needs an additional universal look-up table of size $o(n)$. That is,  we can keep many instances of our data structure for different point sets and  all instances can use the same look-up table.
 \begin{lemma}
   \label{lemma:small}
If a set $P$ contains $d=O(\log n)$ points, then we can keep $P$ in a linear-space data structure $D(P)$ that answers two-dimensional range reporting queries in $O(k)$ time. This data structure relies on a universal look-up table of size $o(n)$.
 \end{lemma}
 \begin{proof}
   First we observe that we can answer a query on a set $P'$ that contains at most $d'=(1/4)\log n/\log\log n$    points using a look-up table of size $o(n)$. Suppose that all points in $P'$ have positive integer coordinates bounded by $d'$. There are $2^{d'\log d'}$ combinatorially different sets $P'$.  For every instance of $P'$, we can ask $(d')^4$ different queries and the answer to each query consists of $O(d')$ points. Hence the total space needed to keep answers to all possible queries on all instances of $P'$ is $O(2^{(\log d')d'}(d')^5)=o(n)$ points. The general case (when point coordinates are arbitrary integers) can be reduced to the case when point coordinates are bounded by $d'$ using reduction to rank space~\cite{GabowBT84,AlstrupBR00}.

A query on $P$ can be reduced to $O(1)$ queries on sets that contain $O(d')$ points using the grid approach~\cite{ChanLP11,AlstrupBR00}. The set of points 
$P$ is divided into $4\log d$ columns $C_i$ and $4\log d$ rows $R_j$ so that every row and every column contains $(1/4)d/\log d$ points. Hence we can support range reporting queries on points in a row/column using the look-up table approach described above. The top set $P_t$ contains a meta-point $(i,j)$ iff the intersection of the $i$-th column and the $j$-th row is not empty, $R_j\cap C_i\not=\emptyset$. Since $P_t$ contains $O(\log^2 d)=o(d')$ points, we can also support queries on $P_t$ in $O(k)$ time. For each meta-point $(i,j)$ in $P_t$ we store the list of points $L_{ij}$ contained in the intersection of the $i$-th column and the $j$-th row, $L_{ij}=C_i\cap R_j\cap P$. 

Consider a query $Q=[a,b]\times [c,d]$. If $Q$ is contained in one column or one row, we answer the query using the data structure for that column/row. Otherwise we identify the rows $R_l$ 
and $R_t$ that contain $c$ and $d$ respectively (i.e., the line $y=c$ is contained in $R_b$ and the line $y=d$ is contained in $R_t$). We also identify the columns $C_f$ and $C_r$ containing $a$ and $b$. We report all points in $Q\cap C_l$, $Q\cap C_r$, $Q\cap R_b$ and $Q\cap R_t$. We find  all meta-points $(i,j)$ in $P_t$ such that $f< i < r$ and $l<j <t$; for every found $(i,j)$ we report all points in $L_{ij}$. 
 \end{proof}


\section{Rectangle Stabbing in Three Dimensions}
\label{sec:stab3d}
We consider the scenario when a set of fat three-dimensional rectangles is stored in a data structure. Corners of rectangles lie on an integer grid of size $U$. Given a query point $q$ with integer coordinates, we must report all rectangles that contain $q$. 

Our construction is based on an uncompressed octree $T$.  The set $P(u)$ for an octree node $u$ contains all rectangles $R$ such that: (i) $R$ contains at least one corner  of $\cell(u)$, but (ii) $R$ does not contain any corners of $\cell(w)$ for any ancestor $w$ of $u$.  We will say that a node $u$ is \emph{relevant} for a rectangle $R$  if $R\in P(u)$. 
\begin{lemma}
  \label{lemma:rect}
  Every fat rectangle $R$ is stored in $O(1)$ sets $P(u)$.
\end{lemma}
\begin{proof}
 Any fat three-dimensional rectangle $R$ can be divided into $O(1)$ cubes $Q_1$, $Q_2$, $\ldots$, $Q_f$. If $R$ satisfies conditions (i) and (ii) with respect to some node $u$, then 
there is at least one  cube $Q_i$ that satisfies conditions (i) and (ii) for some descendant 
$v$ of $u$. Therefore it is sufficient to show that any cube $Q_i$ is relevant for $O(1)$ nodes of $T$.

Let $s$ denote the size of a cube $Q_i$. There is exactly one cell size $\ell$, such that $s< \ell \le 2s$. 
Suppose that a cube $Q_i$ is entirely contained in a size-$\ell$ cell of some node $u$. Since $s\ge \ell/2$, $Q_i$ contains the common corner $\nu$ of all $u$'s children. $Q_i$ does not contain  corners of $u$ or any of its ancestors. Hence $Q_i$ is relevant for at most eight children $u_i$ of $u$. Now suppose that $Q_i$ contains a corner $\nu$ of some size-$\ell$ cell,  $\cell(u_1')$. Since $s< \ell$, $Q_i$ intersects  at most eight cells of size $\ell$ that share the common corner $\nu$. Every cell $\cell(u_i')$ with corner $\nu$ that intersects $Q_i$ satisfies condition (i). Among all ancestors of $u'_i$ there is exactly one node that satisfies both conditions (i) and (ii). Hence every $Q_i$ is relevant for at most eight nodes. 
\end{proof}

For each node $u$ we keep rectangles $R\in P(u)$ in data structures that answer three-dimensional dominance queries. There is one data structure for every corner $\nu$ of $\cell(u)$. If a rectangle $R$ is relevant for a node $u$, then $\cell(u)$ contains one corner $\mu_R$ of $R$\footnote{To avoid tedious details we assume that every rectangle contains exactly one corner of $\cell(u)$. The case when a rectangle $R\in P(u)$ contains two or four corners of $\cell(u)$ can be handled in a similar way.}.  If a rectangle $R\in P(u)$ contains a corner $\nu$ of $\cell(u)$, then we store  $R$ in the data structure $D_{\nu}(u)$ that supports stabbing queries for points $q\in \cell(u)$. Since rectangles $R\in D_{\nu}$ contain the corner $\nu$ of $\cell(u)$, the rectangle stabbing query is equivalent to a three-dimensional dominance query in this case.  Suppose, for example,  that $\nu$ is the corner with the smallest $x$-, $y$-, and $z$-coordinates in $\cell(u)$. Then reporting rectangles $R\in D_{\nu}$ containing the point $q$ is equivalent to reporting corners $\mu_R$ such that $x(\mu_r)\ge x(q)$, $y(\mu_r)\ge y(q)$, and $z(\mu_r)\ge z(q)$. See Figure~\ref{fig:3dstab} for an example in the 2-d case. The linear-space data structure of Chan~\cite{Chan13} answers three-dimensional point reporting queries in $O(\log\log U + k)$ time. Hence $D_{\nu}$ supports stabbing queries for points $q\in \cell(u)$ in $O(\log\log U + k)$ time. 
\begin{figure}[tb]
  \centering
  \includegraphics[width=.3\textwidth]{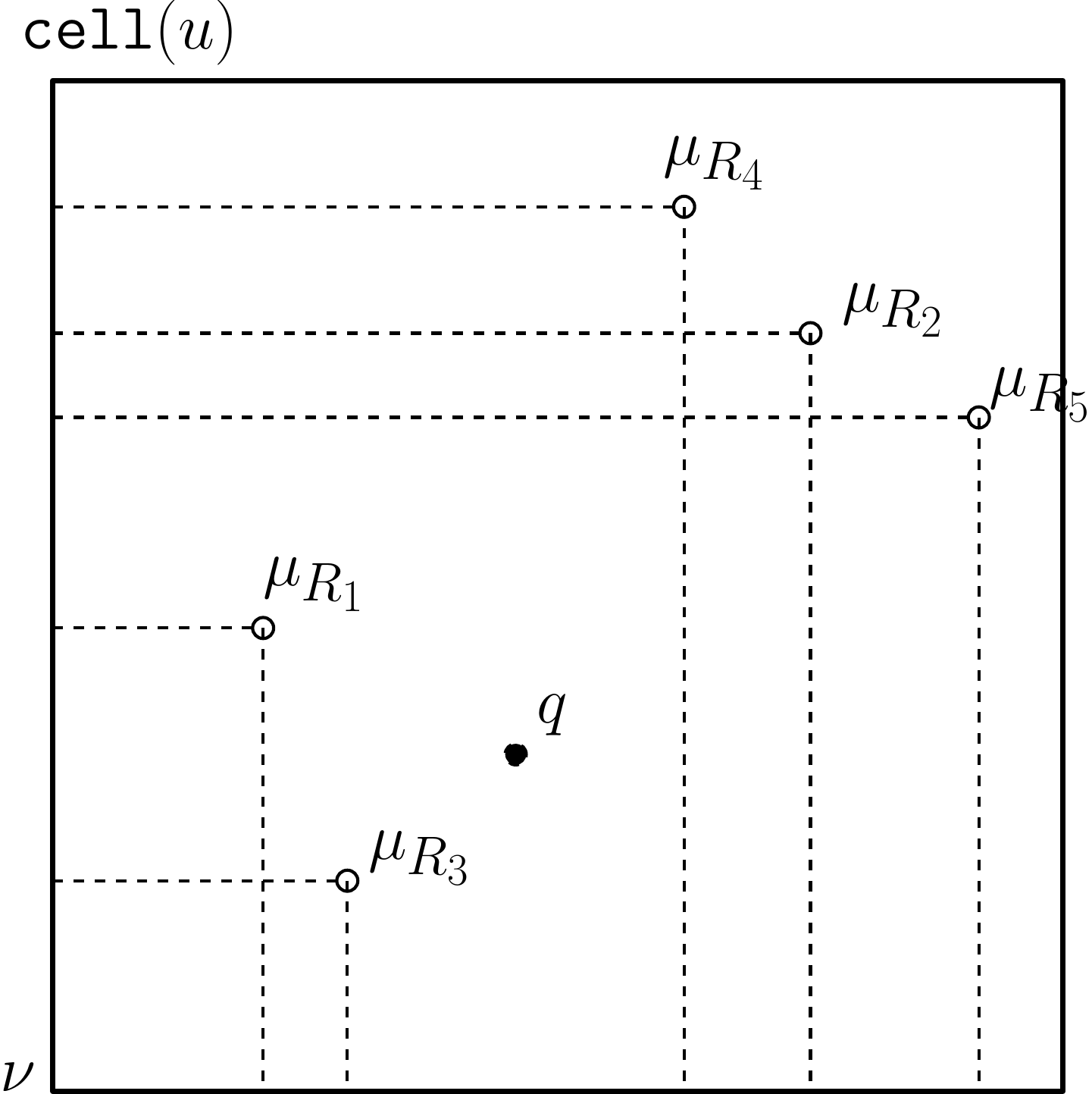}
\caption{\label{fig:3dstab} Relevant rectangles for a quadtree in two dimensions. Relevant rectangles containing the corner $\nu$ are shown with dashed lines, cell boundaries are shown with solid lines. Point $q$ is stabbed by $R_2$, $R_4$, and $R_5$ because $x(\mu_{R_j})\ge x(q)$ and $y(\mu_{R_j})\ge y(q)$ for $j=2,4,5$.} 
\end{figure}

Given a query point $q$, we visit all nodes $u$ on the path from the root to a leaf node that contains $q$. In every visited node $u$ we answer eight dominance queries and thus report all rectangles $R\in P(u)$ that contain $q$. The total time to answer a query is $O(\log U\log \log U+ k)$. Every rectangle stabbed by $q$ is relevant for some  visited node $u$. Therefore our procedure correctly reports all rectangles stabbed by $q$. The data structure uses space $O(n)$ because each rectangle is stored a constant number of times.
\begin{theorem}
\label{theor:stab3d}
  There is an $O(n)$-space data structure that answers three-dimensional rectangle stabbing queries for a set of fat rectangles on a $[U]^3$ grid. Queries are supported in $O(\log U\log\log U+ k)$ time. 
\end{theorem}

\end{document}